\documentclass[
  aps,
  pra,
  reprint,
  amsmath,amssymb,
]{revtex4-2}

\usepackage{amsthm}
\usepackage{graphicx}
\usepackage{dcolumn}
\usepackage{bm}

\usepackage{tikz}
\usetikzlibrary{quantikz2}

\begin{document}

\title{Weak Values Beyond the Weak Limit: \\Measurement-Strength Invariance Under Post-Selection}

\author{Cosmin Andrei}
\author{Vlatko Vedral}
\affiliation{Clarendon Laboratory, University of Oxford,
Parks Road, Oxford OX1 3PU, United Kingdom}

\date{\today}

\begin{abstract}
Weak values are conventionally extracted in the limit of vanishing measurement strength. At finite strength, the corresponding post-selected conditional averages generally acquire nonlinear corrections and need not retain their weak value interpretation. Here we identify a class of measurement protocols for which this expectation fails: the conditional value reconstructed from the measuring probe is exactly independent of the probe strength and therefore coincides with the real part of the weak value. We first analyze this phenomenon in the measurement--disturbance protocol introduced by Lund and Wiseman. We then establish conditions under which this invariance persists, showing that it is determined by the interplay between the probe interaction, a subsequent measurement interaction with an ancillary apparatus, and the final post-selection. These results demonstrate that weak coupling, although generally sufficient for accessing weak values, is not always necessary.
\end{abstract}

\maketitle

\section{Introduction}

Since their introduction by Aharonov, Albert and Vaidman in 1988 \cite{Aharonov1988}, weak values have become a versatile tool in quantum measurement theory and quantum foundations. Originally proposed as the outcome of weak measurements performed on pre- and post-selected quantum ensembles, they have since found applications ranging from weak-value amplification and quantum metrology to quantum state tomography and direct state reconstruction \cite{Dressel2014}. More recently, weak values have attracted considerable interest as operationally accessible conditional quasi-probabilities, which relax some of the constraints of ordinary probability distributions~\cite{ArvidssonShukur2024}, providing an alternative statistical description of quantum systems.

The connection between weak values and conditional quasi-probability distributions follows naturally from the weak values of projection operators. Given a complete set of orthogonal projectors $\{\Pi_j\}$ associated with an observable $A=\sum_j a_j\Pi(a_j)$, the weak value of each projector represents a conditional quasi-probability. These quantities satisfy normalization and allow expectation values to be reconstructed through conditional averaging, closely resembling ordinary conditional probability distributions. Unlike classical probabilities, however, quasi-probabilities are not constrained to lie within the interval $[0,1]$, allowing them to assume negative values or values exceeding unity. Thus, weak values are placed within the broader family of quantum quasi-probability distributions, alongside the Kirkwood--Dirac distribution and related phase-space representations \cite{Wiseman2002, ArvidssonShukur2024}.

Post-selection plays a central role in these nonclassical features. The weak value of an observable,
\begin{equation} 
A_w=\frac{\langle\phi|A|\psi\rangle}{\langle\phi|\psi\rangle},
\label{eq:weak_value}
\end{equation}

can become arbitrarily large whenever the overlap between the pre-selected state $|\psi\rangle$ and the post-selected state $|\phi\rangle$ becomes sufficiently small, provided that $\langle\phi|\psi\rangle \neq 0$. Such anomalous weak values may therefore lie well outside the eigenvalue spectrum of the observable. These anomalous values have attracted considerable attention also because they have been linked to genuinely quantum phenomena such as contextuality and quantum nonclassicality \cite{Pusey2014, dressel2015weak}. Consequently, weak values have become an important theoretical and experimental tool for probing foundational aspects of quantum mechanics.

This quasi-probabilistic interpretation has motivated numerous experimental applications. Weak values have been employed to reconstruct average trajectories in Bohmian mechanics \cite{Wiseman2002,Kocsis2011,Mahler2016,Flack2017,Coffey2011}, to investigate measurement-induced disturbance through Heisenberg measurement-disturbance relations \cite{Lund2010,Rozema2012, ozawa2003universally}, and more recently to explore quantum nonlocal correlations in Bell-type scenarios~\cite{cohen2026quantum}. At the same time, the physical and ontological interpretations of these experiments remain the subject of active debate within the quantum foundations community \cite{Vaidman2017,Barandes2026}. The present work does not address these interpretational questions directly; rather, it focuses on the operational assumptions underlying the extraction of weak values.

Weak values are conventionally reconstructed experimentally through weak measurements combined with a suitable post-selection. The interaction between the quantum system, hereafter referred to simply as the system, and the measuring probe is assumed to be sufficiently weak that the back-action on the quantum state is negligible to leading order. Under this approximation, the conditional shift of the probe's pointer is proportional to the real part of the weak value, providing an experimentally accessible quantity while only minimally disturbing the system. Consequently, the operational reconstruction of weak values is generally regarded as fundamentally tied to the weak-coupling limit, with finite-strength measurements expected to introduce nonlinear corrections to the conditional shift. 

There are, however, simple cases in which this expectation does not hold. Consider a post-selected state that is an eigenstate of the observable whose weak value is being evaluated. In this case, the weak value reduces to the corresponding eigenvalue,
\begin{equation}
  \frac{\langle a|A|\psi \rangle}{\langle a| \psi \rangle} = a.  
  \label{eq:eigenstate_weak_value}
\end{equation}
 As we show below, this elementary case also has an exact finite-strength counterpart: when the post-selection is performed directly onto an eigenstate of the probed observable, the conditional pointer displacement is independent of the probe strength. This provides the simplest example of the measurement-strength invariance studied in this work.

The Lund--Wiseman measurement--disturbance protocol~\cite{Lund2010} exhibits a less immediate form of the same phenomenon. In this protocol, an initial probing is followed by a second measurement apparatus interaction before the final post-selection is performed. Remarkably, in the two-dimensional Lund--Wiseman model, the quantities reconstructed from the probe remain independent of its measurement strength despite the second interaction. The invariance therefore cannot be explained by the elementary eigenstate post-selection argument alone. We investigate its origin by generalizing the protocol to continuous-variable pointers, thereby separating the effect from features specific to qubit meters. We find that the invariance is governed by the combined structure of the probe interaction, the subsequent apparatus interaction, the final post-selection, and the system size.

\section{Measurement-Strength Invariance in the Lund--Wiseman Protocol}

\subsection{Weak-Valued Conditional Probability}
Let
$A=\sum_j a_j\Pi(a_j)$ be an observable with spectral projectors $\Pi(a_j)$. 
Given an initial state $|\psi\rangle$ and a post-selected state $|\phi\rangle$, the
real part of the weak value of $A$ is
\begin{equation}
    {}_{\phi}\langle A\rangle^{\rm weak}_{\psi}
    =
    {\rm Re}
    \left(
    \frac{\langle\phi|A|\psi\rangle}
    {\langle\phi|\psi\rangle}
    \right).
    \label{eq:weak_value_simple}
\end{equation}
Throughout this work, we restrict attention to the real part of the weak value,
which is directly associated with the conditioned pointer displacement. For a general complex initial pointer wavefunction, the same
displacement may also contain a contribution proportional to the imaginary
part of the weak value. We exclude this case by restricting the initial
pointer wavefunctions to be real~\cite{dressel2012significance}. Accordingly, unless
stated otherwise, all references to weak values in what follows refer to their
real parts. 

Using the spectral decomposition of $A$, we note

\begin{equation}
    {}_{\phi}\langle A\rangle^{\rm weak}_{\psi}
    =
    \sum_j a_j\,
    {}_{\phi}\langle \Pi(a_j)\rangle^{\rm weak}_{\psi}.
    \label{eq:weak_value_projector_decomposition}
\end{equation}

This motivates the identification
\begin{equation}
    P^{WV}(a_j|\phi)
    \equiv
    {}_{\phi}\langle \Pi(a_j)\rangle^{\rm weak}_{\psi},
    \label{eq:weak_projector_probability}
\end{equation}

where $P^{WV}(a_j|\phi)$ is called the \textit{weak-valued conditional
probability}. These quantities resemble ordinary conditional probabilities in
that they normalize over a complete set of alternatives,
\begin{equation}
    \sum_j P^{WV}(a_j|\phi)=1.
\end{equation}

However, unlike classical probabilities, weak-valued probabilities need not lie
between $0$ and $1$. They may be negative or larger than unity, which is why
they are more properly understood as conditional quasi-probabilities
\cite{Wiseman2002,Dressel2014}. The notion of a
weak-valued probability distribution is meaningful because the weak value of
$A$ decomposes into the weak values of the projectors $\Pi(a_j)$, as in
Eq.~\eqref{eq:weak_value_projector_decomposition}.

Operationally, weak values are extracted from the statistics of a weak
measurement followed by post-selection. To see this, consider a standard von
Neumann interaction between the system and a measuring probe,
\begin{equation}
    H_{\rm int}=g(t)\,A\otimes p,
\end{equation}

where $g(t)$ is the interaction strength and $p$ is the momentum conjugate to the probe's
pointer position $x$. Neglecting the free evolution during the interaction and defining the
integrated coupling strength $g \equiv \int dt\, g(t)$,
the overall unitary becomes $
U_g = e^{-ig A\otimes p}$. Throughout this paper, we will take $\hbar=1$ and always assume a nonzero probe coupling strength. Assume the composite system is initially in state $|\Psi_{in}\rangle = |\psi\rangle \otimes |\gamma\rangle$, with $|\psi\rangle$ and $|\gamma\rangle$ corresponding to the system and probe respectively. For the real initial probe considered here, if the probe pointer is centered at the origin, then
the expected value of $A$ relates to the average pointer position via $\langle A \rangle =\frac{1}{g}  \int dx\, x\,P_g(x)$ where $P_g(x) = ||(I \otimes \langle x|) U_g |\Psi_{in}\rangle|| ^2 $ is the distribution of pointer values. An important property of the previous expression is that $\langle A \rangle$ is independent of the interaction strength $g$. 

Conditioning on a successful post-selection $|\phi\rangle$, one may define
the conditional expectation value

\begin{equation}
    {}_{\phi}\langle A\rangle_{\psi}^{(g)}
    =
    \frac{1}{g}
    \int dx\, x\,P_g(x|\phi),
    \label{eq:finite_strength_conditional_average}
\end{equation}

where $P_g(x|\phi) = \frac{P_g(x, \phi)}{P_g(\phi)}$ is the post-selected pointer distribution. In the weak-measurement limit, the disturbance caused by the
probe becomes negligible, and this conditional expectation value reduces to the real part of the weak
value,
\begin{equation}
    \lim_{g\rightarrow 0}
    {}_{\phi}\langle A\rangle_{\psi}^{(g)}
    =
    {}_{\phi}\langle A\rangle^{\rm weak}_{\psi}.
    \label{eq:weak_limit_pointer_average}
\end{equation}

Unlike the expected value $\langle A \rangle$, the finite-strength conditional average ${}_{\phi}\langle A\rangle_{\psi}^{(g)}$ is generally not invariant under the interaction strength $g$. In particular, ${}_{\phi}\langle A\rangle_{\psi}^{(g)}$ need not decompose linearly into projector contributions as in Eq.~\eqref{eq:weak_value_projector_decomposition},
so the finite-strength conditional average cannot in general be identified with the standard weak-valued probability distribution. In the limit \(g\to0\), however, the conditional average reduces to the weak value, and the corresponding weak-valued probability interpretation is recovered.

\subsection{Measurement-Disturbance Relationship}
An important application of weak-valued probability distributions was introduced by Lund and Wiseman~\cite{Lund2010}, who showed that they provide an operational means of reconstructing Ozawa's measurement precision and disturbance~\cite{ozawa2003universally}. Because our investigation is motivated by this protocol, we briefly summarize its main ingredients before presenting our finite-strength analysis.

In an indirect measurement model, a system in the initial mixed state $\rho$
interacts with a measurement apparatus prepared in the state $\theta$.
The system and apparatus interact through the unitary $U_A^{\rm app}$,
which correlates the system observable $A$ with the apparatus readout
$M_{\rm app}$. Simultaneously, another system observable $B$ may be disturbed by the
same interaction. Ozawa defines the measurement precision and disturbance as
\begin{equation}
    \epsilon(A)=
    \left\langle
    \left[
    {U_A^{\text{app}}}^\dagger(I\otimes M_\text{app})U_A^{\text{app}}-A\otimes I
    \right]^2
    \right\rangle^{1/2},
    \label{eq:ozawa_precision}
\end{equation}

and
\begin{equation}
    \eta(B)=
    \left\langle
    \left[
    {U_A^{\text{app}}}^\dagger(B\otimes I)U_A^{\text{app}}-B\otimes I
    \right]^2
    \right\rangle^{1/2},
    \label{eq:ozawa_disturbance}
\end{equation}

where the expectation value is taken over $\rho\otimes\theta$. These quantities
appear in the universally valid measurement--disturbance relation,
\begin{equation}
    \epsilon(A)\eta(B)+\epsilon(A)\sigma(B)+\sigma(A)\eta(B)
    \geq
    \frac{1}{2}|\langle[A,B]\rangle|.
    \label{eq:Ozawa}
\end{equation}

\begin{figure}[t!]
    \centering
    \begin{quantikz}[row sep=0.9cm, column sep=0.45cm]
        \lstick{$\rho$} & \gate[2]{U^{\text{probe}}_{A/B}} \gategroup[2,steps=2,style={dashed,rounded corners,fill=blue!20, inner xsep=2pt},background,label style={label position=below,anchor=north,yshift=-0.2cm}]{{ Probe}} & \qw & \qw & \qw & & \gate[2]{U_A^{\text{app}}} \gategroup[2,steps=2,style={dashed,rounded corners,fill=red!20, inner xsep=2pt},background,label style={label position=below,anchor=north,yshift=-0.2cm}]{{Measurement Apparatus}} & &\meter{B}\\
        \lstick{$\gamma$}  & \targ{}  & \meter{M_{\text{probe}}}&\wireoverride{n}& \wireoverride{n} &\wireoverride{n} \lstick{$\theta$} & \targ{}& \meter{M_\text{app}} 
    \end{quantikz}
    \caption{
Generalized schematic of the Lund--Wiseman measurement--disturbance protocol~\cite{Lund2010}. The system first interacts with a weakly-coupled probe through the unitary $U^{\mathrm{probe}}_{A/B}$, which performs a weak measurement of either the observable $A$ (precision) or $B$ (disturbance). The system subsequently interacts with the measurement apparatus through the unitary $U_A^{\text{app}}$. For the precision, the final readout is performed on the apparatus meter $M_{\rm app}$, whereas for the disturbance the final readout is performed directly on the system in the $B$ basis. Adapted from Ref.~\cite{Lund2010}.}
    \label{fig:LundProtocol}
\end{figure}
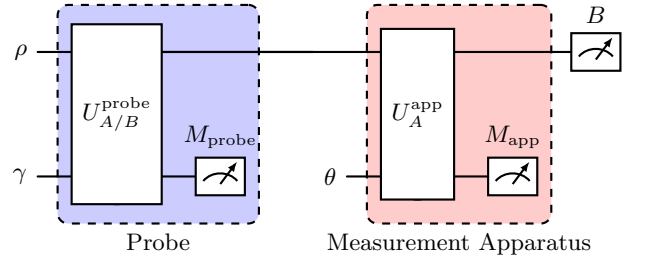

Here $\sigma(A)$ and $\sigma(B)$ denote the standard deviations of $A$ and $B$
in the initial system state. For the present paper, however,
Eq.~\eqref{eq:Ozawa} is only motivational. We will not study the validity of
the measurement--disturbance relation itself as this was done elsewhere~\cite{Rozema2012, baek2013experimental, erhart2012experimental}. Instead, we focus on the
weak-valued probability distributions used in the Lund--Wiseman reconstruction.

The key observation of Lund and Wiseman is that the quantities in
Eqs.~\eqref{eq:ozawa_precision} and~\eqref{eq:ozawa_disturbance} can be written
as root-mean-square differences over weak-valued joint distributions. For the precision, one weakly probes the initial value $a_i$ of $A$ before the apparatus
interaction and post-selects on the final apparatus readout $a_f$. For the 
disturbance, one weakly probes the initial value $b_i$ of $B$ before the apparatus
interaction and post-selects on the final system readout $b_f$. The operational procedure is illustrated in Fig.~\ref{fig:LundProtocol}.

In both the precision and disturbance experiments, a weakly-coupled probe is inserted
immediately before the measurement apparatus in order to extract information
about the pre-interaction value of the relevant observable through the unitaries $U_A^{\text{probe}}$ or $U_B^{\text{probe}}$. The system then interacts with the measurement
apparatus through the unitary $U_A^{\text{app}}$. For the precision measurement, the final
readout is performed on the apparatus meter, whereas for the disturbance
measurement the final readout is performed directly on the system. By combining
the probe outcomes with the appropriate post-selection, Lund and Wiseman
showed that Ozawa's precision and disturbance can be reconstructed from
weak-valued probability distributions via

\begin{equation}
    \epsilon(A)^2
    =
    \sum_{a_i,a_f}
    (a_i-a_f)^2 P^{WV}(a_i,a_f),
    \label{eq:epsilon_joint}
\end{equation}

and

\begin{equation}
    \eta(B)^2
    =
    \sum_{b_i,b_f}
    (b_i-b_f)^2 P^{WV}(b_i,b_f),
    \label{eq:eta_joint}
\end{equation}

where $P^{WV}(a_i,a_f) = P^{WV}(a_i|a_f) P(a_f)$ is defined as the weak-valued joint probability distribution. 

In the present paper, Eqs.~\eqref{eq:epsilon_joint} and~\eqref{eq:eta_joint}
are important not because we wish to study the numerical values of
$\epsilon(A)$ and $\eta(B)$, but because they identify the
weak-valued probabilities as the operational quantities from which precision
and disturbance are reconstructed. Our central question is whether these
weak-valued probabilities genuinely require a weakly-coupled probe, or whether they can
sometimes be obtained exactly using a finite-strength probing interaction.

In the Lund--Wiseman setup the post-selection is not performed immediately after the probe. Instead, the system interacts with the measurement apparatus through the unitary $U_A^{\text{app}}$, and the post-selection is performed only afterwards. One must therefore use the generalized weak-value formula for a mixed initial state, an intermediate unitary evolution, and a final post-selection~\cite{Wiseman2002, Lund2010}. For the disturbance of $B$, the weak-valued probability that the initial value was $b_i$, conditioned on the final value $b_f$, is
\begin{equation}
\begin{aligned}
&P^{WV}(b_i|b_f)
= {}_{b_f}\langle \Pi(b_i)\rangle^{\rm weak}_{\rho}{} = \\
&\operatorname{Re}\Bigg[
\frac{
\operatorname{Tr}\!\Big[
(\Pi(b_f)\otimes I)\,
U_A^{\text{app}}
(\Pi(b_i)\otimes I)
(\rho\otimes\theta)
{U_A^{\text{app}}}^\dagger
\Big]}{
\operatorname{Tr}\!\Big[
(\Pi(b_f)\otimes I)\,
U_A^{\text{app}}
(\rho\otimes\theta)
{U_A^{\text{app}}}^\dagger
\Big]
}
\Bigg].
\end{aligned}
\label{eq:Lund_disturbance_weak_value}
\end{equation}

Similarly, for the precision of $A$, the weak-valued probability that the initial system value was $a_i$, conditioned on the final apparatus readout $a_f$, is
\begin{equation}
\begin{aligned}
&P^{WV}(a_i|a_f)
= {}_{a_f}\langle \Pi(a_i)\rangle^{\rm weak}_{\rho}{} = \\
&{\rm Re}\!\Bigg[
\frac{
{\rm Tr}\!\left[
(I\otimes\Pi(a_f))\,
U_A^{\text{app}}\,
(\Pi(a_i)\otimes I)\,
(\rho\otimes\theta)\,
{U_A^{\text{app}}}^\dagger
\right]
}{
{\rm Tr}\!\left[
(I\otimes\Pi(a_f))\,
U_A^{\text{app}}\,
(\rho\otimes\theta)\,
{U_A^{\text{app}}}^\dagger
\right]
}
\Bigg].
\end{aligned}
\label{eq:Lund_precision_weak_value}
\end{equation}

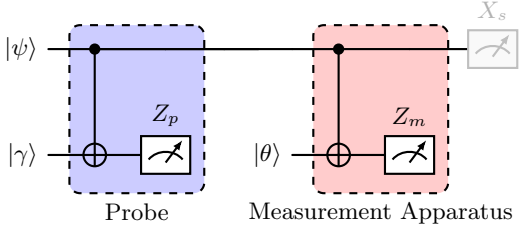
\begin{figure}[t!] 
\centering 
\begin{quantikz}[row sep=0.9cm, column sep=0.45cm] 
\lstick{$|\psi\rangle$} & \ctrl{1} \gategroup[2,steps=2,style={dashed,rounded corners,fill=blue!20, inner xsep=2pt},background,label style={label position=below,anchor=north,yshift=-0.2cm}]{{Probe}} & \qw & \qw & \qw & & \ctrl{1} \gategroup[2,steps=2,style={dashed,rounded corners,fill=red!20, inner xsep=2pt},background,label style={label position=below,anchor=north,yshift=-0.2cm}]{{Measurement Apparatus}} & &\meter[style={ draw=gray!50, text=gray!60, fill=gray!5}, label style={text=gray!60}]{X_s}\\ 
\lstick{$|\gamma\rangle$} & \targ{} & \meter{Z_p} &\wireoverride{n}& \wireoverride{n} &\wireoverride{n} \lstick{\ket{\theta}} & \targ{}& \meter{Z_m} \end{quantikz} 
 \caption{Three-qubit implementation of the Lund--Wiseman precision protocol~\cite{Lund2010}. The system first interacts with the probe (P) through a CNOT gate, implementing a variable-strength probe of $Z$, and subsequently interacts with the measurement apparatus (MA) through a second CNOT gate. The probe and measurement apparatus are read out in the $Z$ basis, yielding outcomes $Z_p$ and $Z_m$, respectively. Figure adapted from Ref.~\cite{Lund2010}.}
   \label{fig:LWqubit_precision}

\end{figure}

\subsection{Qubit Example}

We now specialize to the qubit example proposed in Ref.~\cite{Lund2010}, with $A=Z$ and $B=X$. We first focus on the precision protocol. The probe and measurement apparatus are both read out in the $Z$ basis, with outcomes denoted by $Z_p=\pm1$ and $Z_m=\pm1$, respectively. Throughout the following explicit calculations, we adopt the tensor-product ordering $\mathcal{H}_S \otimes \mathcal{H}_{\rm P} \otimes \mathcal{H}_{\rm MA}$
corresponding to the system, probe, and measurement apparatus, respectively. When only a two-factor tensor product is written, the second Hilbert space refers to either the probe or the measurement apparatus, as clear from context.

The system is initially prepared in the state
$|\psi\rangle=\alpha|0\rangle+\beta|1\rangle$, the probe in the state
$|\gamma\rangle=\gamma|0\rangle+\bar{\gamma}|1\rangle$, and the measurement apparatus in the state
$|\theta\rangle=\cos\theta\,|0\rangle+\sin\theta\,|1\rangle$, where
$\gamma,\bar{\gamma}\in\mathbb{R}^{+}$ satisfy
$\gamma^{2}+\bar{\gamma}^{2}=1$. The quantum circuit is shown in Fig.~\ref{fig:LWqubit_precision}.

The probe and the measurement apparatus each interact with the system in the $Z$ basis through CNOT gates, described by the unitaries $U^{\rm probe}_{Z}$ and
$U^{\rm app}_{Z}$, respectively. The probe strength is quantified by
$s_{\gamma}=2|\gamma|^{2}-1$, while the measurement apparatus strength is quantified by
$s_{\theta}=\cos(2\theta)$. In the standard weak-value formalism, the probe is assumed to operate in the weak-coupling regime, $|s_{\gamma}|\ll1$, whereas the apparatus strength $s_{\theta}$ is allowed to vary continuously from a projective measurement ($\theta=0$) to the no-measurement limit ($\theta=\pi/4$).

Reading out the probe in the $Z$ basis realizes the two-outcome POVM

\begin{equation}
    E_{\pm}
    =
    \frac{1}{2}
    \left(
    I\pm   s_\gamma Z
    \right).
    \label{eq:probe_povm_Z}
\end{equation}

Hence, the expectation value of $Z$ inferred from the probe statistics is

\begin{equation}
    \langle Z\rangle
    =
    \frac{
    P(Z_p=+1)-P(Z_p=-1)
    }{
    s_\gamma
    }.
    \label{eq:probe_average_Z}
\end{equation}

Conditioning on a final apparatus outcome $Z_m$, the conditional expectation value is

\begin{equation}
    {}_{Z_m}\langle Z\rangle^{(s_\gamma)}_{\psi}
    =
    \frac{
    P(Z_p=+1|Z_m)-P(Z_p=-1|Z_m)
    }{
      s_\gamma
    }.
    \label{eq:conditional_probe_average_Z}
\end{equation}

The standard interpretation of Eq.~\eqref{eq:conditional_probe_average_Z} is that it coincides with the weak value only in the limit $s_\gamma\rightarrow0$. 
Since $Z$ has only two eigenvalues, the corresponding weak-valued probabilities are reconstructed through

\begin{equation}
    2P^{WV}(Z_i=\pm1|Z_m)
    =
    1
    \pm
    {}_{Z_m}\langle Z\rangle^{\rm weak}_{\psi}.
    \label{eq:weak_prob_from_Z}
\end{equation}

At finite strength, one would normally expect corrections because the probe is no longer negligibly disturbing. We now compute the relevant conditional probabilities exactly, without taking the weak limit.

After the probe and apparatus interactions, the joint and marginal probabilities for obtaining probe outcome $Z_p$ and apparatus outcome $Z_m$ are
\begin{widetext}
\begin{equation}
\begin{aligned}
    P(Z_p,Z_m)
    = 
   \left\|
    I\otimes \langle Z_p|\otimes \langle Z_m|
    \left(
    U^{\rm app}_{\rm Z}
    U^{\rm probe}_{\rm Z}
    \right)
    |\psi\rangle\otimes|\gamma\rangle\otimes|\theta\rangle
    \right\|^2
    \label{eq:joint_probability_exact}
    \end{aligned}
\end{equation}
\begin{equation}
\begin{aligned}
    P(Z_m)
    = 
    \left\|
    I\otimes I\otimes \langle Z_m|
    \left(
    U^{\rm app}_{\rm Z}
    U^{\rm probe}_{\rm Z}
    \right)
    |\psi\rangle\otimes|\gamma\rangle\otimes|\theta\rangle
    \right\|^2.
    \label{eq:marginal_probability_exact}
\end{aligned}
\end{equation}
\end{widetext}

The conditional probability appearing in Eq.~\eqref{eq:conditional_probe_average_Z} is then $P(Z_p|Z_m)=\frac{P(Z_p,Z_m)}{P(Z_m)}$. Evaluating these probabilities explicitly gives, for the post-selection $Z_m=+1$,

\begin{equation}
    P(Z_p=+1,Z_m=+1)
    =
    |\alpha|^2|\gamma|^2\cos^2\theta
    +
    |\beta|^2|\bar{\gamma}|^2\sin^2\theta,
    \label{eq:P_plus_plus}
\end{equation}

and

\begin{equation}
    P(Z_p=-1,Z_m=+1)
    =
    |\alpha|^2|\bar{\gamma}|^2\cos^2\theta
    +
    |\beta|^2|\gamma|^2\sin^2\theta.
    \label{eq:P_minus_plus}
\end{equation}

The marginal probability is therefore

\begin{equation}
    P(Z_m=+1)
    =
    |\alpha|^2\cos^2\theta
    +
    |\beta|^2\sin^2\theta.
    \label{eq:P_Zm_plus}
\end{equation}

Substituting Eqs.~\eqref{eq:P_plus_plus}--~\eqref{eq:P_Zm_plus} into Eq.~\eqref{eq:conditional_probe_average_Z}, we find

\begin{equation}
    {}_{Z_m=+1}\langle Z\rangle ^{(s_\gamma)}_{\psi}
    =
    \frac{
    |\alpha|^2\cos^2\theta
    -
    |\beta|^2\sin^2\theta
    }{
    |\alpha|^2\cos^2\theta
    +
    |\beta|^2\sin^2\theta
    }.
    \label{eq:finite_strength_Z_value}
\end{equation}

Similarly, for the post-selection $Z_m=-1$ one obtains

\begin{equation}
    {}_{Z_m=-1}\langle Z\rangle^{(s_\gamma)}_{\psi}
    =
    \frac{
    |\alpha|^2\sin^2\theta
    -
    |\beta|^2\cos^2\theta
    }{
    |\alpha|^2\sin^2\theta
    +
    |\beta|^2\cos^2\theta
    }.
    \label{eq:finite_strength_Z_value_minus}
\end{equation}

Eqs.~\eqref{eq:finite_strength_Z_value} and~\eqref{eq:finite_strength_Z_value_minus} explicitly demonstrate measurement-strength invariance in the precision protocol. In both cases, the dependence on the probe strength $s_\gamma$
cancels exactly. The resulting finite-strength conditional values therefore coincide with their weak-limit values,
\begin{equation}
\begin{aligned}
    {}_{+1}\langle Z\rangle ^{(s_\gamma)}_{\psi} &= {}_{+1}\langle Z\rangle ^{\rm weak}_{\psi} \\ 
    {}_{-1}\langle Z\rangle ^{(s_\gamma)}_{\psi} &= {}_{-1}\langle Z\rangle ^{\rm weak}_{\psi}.
\end{aligned}
\end{equation}
 
Consequently, the weak-valued probabilities reconstructed through Eq.~\eqref{eq:weak_prob_from_Z} are recovered exactly at finite probe strength and are independent of the probe strength:

\begin{equation}
    P^{WV}(Z_i=\pm1|Z_m)
    =
    \frac{1}{2}
    \left(
    1
    \pm
    {}_{Z_m}\langle Z\rangle^{(s_\gamma)}_{\psi}
    \right).
    \label{eq:invariance}
\end{equation}

Therefore, in the qubit precision protocol, the weak-valued probability distribution can be reconstructed exactly from finite-strength probe statistics, without taking the weak-measurement limit. The weak-coupling limit is not necessary for reconstructing these quantities.

As previously discussed, the expectation value $\langle Z \rangle$ can be recovered from the probe statistics at any strength. The nontrivial fact is that the post-selected conditional quantity ${}_{Z_m}\langle Z\rangle^{(s_\gamma)}_{\psi}$, which would normally be expected to acquire nonlinear finite-strength corrections as suggested by Eq.~\eqref{eq:finite_strength_conditional_average}, also remains strength invariant. Since the weak-valued probability distribution is built from this conditional value, it inherits the same invariance.  If the finite-strength conditional value agrees exactly with the weak value, then the linear weak-value decomposition of Eq.~\eqref{eq:weak_value_projector_decomposition} is preserved. In this sense, the more fundamental result is the equality between the finite-strength conditional value and the weak value; the strength invariance of the reconstructed weak-valued probability distribution is then a consequence.

The corresponding analysis applies to the disturbance protocol of Lund and Wiseman, in which the probe interaction is arranged to measure the system observable $X$. The probe is subsequently read out in the $Z$ basis, with outcomes $Z_p=\pm1$, from which the initial $X$ values of the system are inferred. Following the measurement-apparatus interaction, the system is post-selected in the $X$ basis, with outcome $X_s=\pm1$, as in Fig.~\ref{fig:LWqubit_disturbance}.

\begin{equation}
{}_{X_s}\langle X\rangle^{(s_{\gamma})}_{\psi}
=
\frac{
P(Z_p=+1|X_s)-P(Z_p=-1|X_s)
}{
s_{\gamma}
},
\label{eq:disturbance_weak_value_qubit}
\end{equation}
and 
\begin{equation}
2P^{WV}(X_i=\pm1|X_s)
=
1 \pm {}_{X_s}\langle X\rangle^{\rm weak}_{\psi}.
\label{eq:disturbance_weak_prob}
\end{equation}

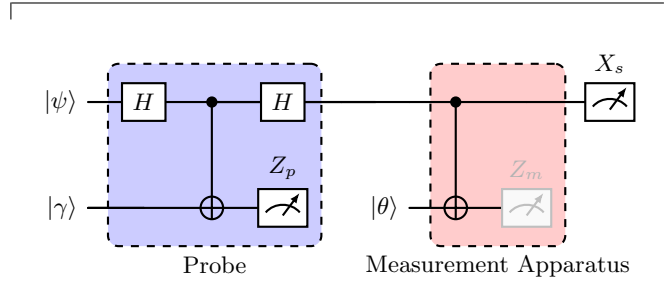
\begin{figure}[t!]
    \centering
    \begin{quantikz}[row sep=0.9cm, column sep=0.45cm]
        \lstick{$|\psi\rangle$} &\gate{H} \gategroup[2,steps=3,style={dashed,rounded corners,fill=blue!20, inner xsep=2pt},background,label style={label position=below,anchor=north,yshift=-0.2cm}]{{Probe}}& \ctrl{1}  & \gate{H} \qw & \qw & \qw & & \ctrl{1} \gategroup[2,steps=2,style={dashed,rounded corners,fill=red!20, inner xsep=2pt},background,label style={label position=below,anchor=north,yshift=-0.2cm}]{{Measurement Apparatus}} & &\meter{X_s}\\
        \lstick{$|\gamma\rangle$} & \qw  & \targ{}  & \meter{Z_p} &\wireoverride{n}& \wireoverride{n} &\wireoverride{n} \lstick{\ket{\theta}} & \targ{}& \meter[style={ draw=gray!50, text=gray!60, fill=gray!5}, label style={text=gray!60}]{Z_m} 
    \end{quantikz}
    \caption{Three-qubit implementation of the Lund--Wiseman disturbance protocol~\cite{Lund2010}. The system first interacts with the probe (P), with Hadamard gates before and after the CNOT implementing a variable-strength probe of $X$. The system subsequently interacts with the measurement apparatus (MA) through a CNOT gate. The probe is read out in the $Z$ basis, yielding $Z_p$, while the final system state is measured in the $X$ basis, yielding $X_s$. Figure adapted from Ref.~\cite{Lund2010}.}
    \label{fig:LWqubit_disturbance}
\end{figure}

 The explicit derivation of the finite-strength conditionals given in the Appendix~\ref{app:finite_strength_disturbance} shows that the extracted quantities ${}_{X_s}\langle X\rangle^{(s_{\gamma})}_{\psi}$ are again independent of the probe strength and coincide exactly with the corresponding weak values ${}_{X_s}\langle X\rangle^{\rm weak}_{\psi}$. Hence, the weak-valued probabilities entering the reconstruction of $\eta(X)$ are invariant as well. Surprisingly, the invariance occurs despite the probe and measurement apparatus coupling through different observables, $X$ and $Z$, respectively.

 The independence of the reconstructed precision and disturbance from the probe strength was noted numerically by Lund and Wiseman in their analysis of imperfect CNOT gates~\cite{Lund2010}, and experimentally by Rozema et al.~\cite{Rozema2012}. Our analysis shows explicitly that, in the ideal qubit model, the cancellation occurs exactly at the level of the finite-strength conditional value.

\section{Conditions for Probe-Strength Invariance}

The simplest origin of measurement-strength invariance can already be seen
from the eigenstate post-selection considered in Eq.~\eqref{eq:eigenstate_weak_value}. Suppose that the
system couples only to the probe through the von Neumann interaction and is subsequently post-selected onto an eigenstate
$|a_f\rangle$ of $A$, with
$A|a_f\rangle=a_f|a_f\rangle$.

 If the probe pointer is
initially centered at the origin, the finite-strength conditional average
defined in Eq.~\eqref{eq:finite_strength_conditional_average} satisfies
\begin{equation}
{}_{a_f}\langle A\rangle_{\psi}^{(g)}
=
\frac{1}{g}
\int
dx\,x\,P_g(x|a_f)
=
a_f.
\label{eq:eigenstate_strength_invariance}
\end{equation}
The equality in Eq.~\eqref{eq:eigenstate_weak_value} therefore has a direct operational counterpart:
post-selection onto an eigenstate of the measured observable makes the
conditional probe pointer displacement exactly independent of the interaction
strength. In this sense, the eigenstates of $A$ form a special class of
post-selected states.

Normally, this invariance is not expected to survive under an arbitrary evolution
between the probe interaction and the final post-selection. To see this,
consider an arbitrary state $|\phi\rangle$. For any chosen post-selected
eigenstate $|a_f\rangle$, one may introduce a unitary $V$ such that
$V|\phi\rangle=|a_f\rangle$. If $V$ is applied before the final projection onto $|a_f\rangle$, the
corresponding weak value becomes
\begin{equation}
\begin{aligned}
{}_{\phi}\langle A\rangle_{\psi}^{\rm weak}
&=
\operatorname{Re}
\left(
\frac{
\langle a_f|VA|\psi\rangle
}{
\langle a_f|V|\psi\rangle
}
\right)
\\
&=
\operatorname{Re}
\left(
\frac{
\langle\phi|A|\psi\rangle
}{
\langle\phi|\psi\rangle
}
\right).
\end{aligned}
\label{eq:rotated_eigenstate_postselection}
\end{equation}
The effective post-selection before the action of $V$
is therefore arbitrary, despite the final projection being performed onto an eigenstate of $A$. In this case, the weak value is no longer constrained to equal $a_f$ and may, in
particular, become anomalous. Correspondingly, the exact finite-strength
conditional pointer average will generally acquire a dependence on $g$.

At first sight, the generalized weak-value expression in Eq.~\eqref{eq:Lund_disturbance_weak_value} appears
to belong to precisely the situation in Eq.~\eqref{eq:rotated_eigenstate_postselection}. The probe is followed by
the measurement-apparatus unitary $U_A^{\text{app}}$, and the final post-selection is
performed only after the intermediate evolution. Nevertheless, we still recover strength invariance as shown in Eq.~\eqref{eq:invariance}. 

The
reason is that its intermediate evolution and post-selection are not
arbitrary. The unitary $U_A^{\text{app}}$ is an entangling measurement interaction between
the system and the measurement apparatus, while the final post-selection is represented by an effect on the joint system–apparatus Hilbert space. For the disturbance
protocol, the post-selection is done on the subspace
$\Pi(b_f)\otimes I$, whereas for the precision protocol it has the form
$I\otimes\Pi(a_f)$. Measurement-strength invariance must therefore be understood as a property of the combined probe interaction, subsequent apparatus interaction, and final post-selection, rather than of the system post-selection alone.

We now investigate the origin of this strength invariance using a more
general measurement model that preserves the structure of the Lund--Wiseman
qubit example while replacing the qubit meters by continuous-variable pointers.
This removes any dependence on properties specific to finite-dimensional
meters and isolates the roles of the entangling interactions and
post-selection. 

Consider two cases in turn: first, \textit{identical entangling observables}
for the system--probe and system--measurement apparatus interactions, corresponding to the precision protocol; and second, \textit{distinct entangling
observables} for the two interactions, corresponding to the
disturbance protocol.

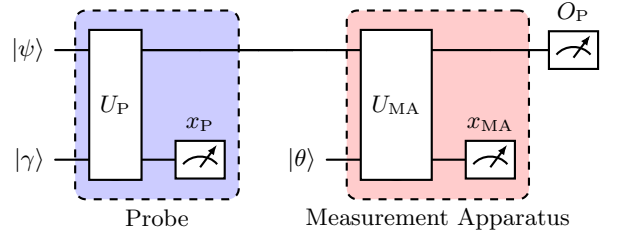
\begin{figure}[t]
    \centering
    \begin{quantikz}[row sep=0.9cm, column sep=0.45cm]
        \lstick{$|\psi\rangle$} & \gate[2]{U_{\rm P}} \gategroup[2,steps=2,style={dashed,rounded corners,fill=blue!20, inner xsep=2pt},background,label style={label position=below,anchor=north,yshift=-0.2cm}]{{Probe}} & \qw & \qw & \qw & & \gate[2]{U_{\rm MA}} \gategroup[2,steps=2,style={dashed,rounded corners,fill=red!20, inner xsep=2pt},background,label style={label position=below,anchor=north,yshift=-0.2cm}]{{Measurement Apparatus}} & &\meter{O_{\rm P}}\\
         \lstick{$|\gamma\rangle$}  & \targ{}  & \meter{x_\text{P}}&\wireoverride{n}& \wireoverride{n} &\wireoverride{n} \lstick{\ket{\theta}} & \targ{}& \meter{x_\text{MA}} 
    \end{quantikz}
   \caption{
The system first interacts with the probe (P) through
$U_{\rm P}=\exp(-ig_p O_{\rm P}\otimes p_{\rm P})$, after which
the conjugate pointer observable $x_{\rm P}$ is read out. The system subsequently interacts
with the measurement apparatus (MA) through
$U_{\rm MA}=\exp(-ig_m O_{\rm MA}\otimes p_{\rm MA})$. The probe statistics are then conditioned either on a final system measurement in the eigenbasis of $O_{\rm P}$ or on an outcome $x_f$ of the MA pointer observable $x_{\rm MA}$, depending
on the post-selection under consideration. The P and MA pointers are initially
described by the position-space wavefunctions $\gamma(x)$ and $\theta(x)$,
respectively.
}
    \label{fig:general_protocol}
\end{figure}
\subsection{Identical Entangling Observables}

Consider an $N$-dimensional system
initially prepared in an arbitrary state $|\psi\rangle$, together with two
independent continuous-variable pointers representing the probe (P)
and the measurement apparatus (MA). Their initial states are described by real, normalized 
position-space wavefunctions $\gamma(x)$ and $\theta(x)$,
respectively,
\begin{equation}
\int dx\,|\gamma(x)|^2
=
\int dx\,|\theta(x)|^2
=
1,
\end{equation}
with vanishing mean positions,
$\langle x_{\rm P}\rangle_\gamma
=
\langle x_{\rm MA}\rangle_\theta
=
0$.

The system first interacts with the probe and subsequently with the
measurement apparatus through the von Neumann interaction unitaries
\begin{equation}
U_{\rm P}
=
e^{-ig_p O_{\rm P}\otimes p_{\rm P}},
\qquad
U_{\rm MA}
=
e^{-ig_m O_{\rm MA}\otimes p_{\rm MA}},
\label{eq:general_unitaries}
\end{equation}
where $O_{\rm P}$ and $O_{\rm MA}$ denote the system observables generating
the two entangling interactions shown in
Fig.~\ref{fig:general_protocol}, while $g_p$ and $g_m$ are the corresponding
coupling strengths. Following the second interaction, the probe pointer
is measured and its statistics are conditioned on a suitable post-selection.

Our objective is to determine under what conditions the post-selected
pointer displacement reconstructed from the probe is exactly
independent of the probing strength $g_p$. As we shall show, the answer is
governed by the post-selected state and the relationship between the entangling observables
$O_{\rm P}$ and $O_{\rm MA}$.

We begin with the simplest case, which mirrors the coupling structure of the
precision protocol considered in the preceding section. We take the
probe and measurement apparatus to interact with the system through the same
observable, $O_{\rm P}
=
O_{\rm MA}
=
A,$ where $A$ is given by the nondegenerate spectral decomposition $A
=
\sum_{j=1}^{N}
a_j |a_j\rangle\langle a_j|$.

Analogously to Eq.~\eqref{eq:finite_strength_conditional_average}, we define the finite-strength conditional value reconstructed
from the probe displacement as
\begin{equation}
{}_{\phi}\langle A\rangle_{\psi}^{(g_p)}
=
\frac{1}{g_p} \int dx\, x\, P_{g_p}(x|\phi)
.
\label{eq:composite_finite_strength_value}
\end{equation}
 with the crucial distinction that the post-selection is performed after the measurement-apparatus interaction. \\

\newtheorem*{lemma}{Lemma}
\begin{lemma}
\label{lem:identical_observables}
Let an $N$-dimensional system initially prepared in
\begin{equation*}
|\psi\rangle
=
\sum_{k=1}^{N}c_k|a_k\rangle
\end{equation*}
interact successively with a probe and a measurement apparatus through
\begin{equation*}
U_{\rm P}
=
e^{-ig_p A\otimes p_{\rm P}},
\qquad
U_{\rm MA}
=
e^{-ig_m A\otimes p_{\rm MA}}.
\end{equation*}
Suppose that, following both interactions, the measurement apparatus is
post-selected on the pointer-position outcome $x_m$, with nonzero
post-selection probability. 

Then the finite-strength conditional value
${}_{x_m}\langle A\rangle_{\psi}^{(g_p)}$ is independent of the probe
strength $g_p$.
\end{lemma}

\begin{proof}

Following the two interactions, the joint state is
\begin{equation}
|\Psi_2\rangle
=
\sum_{k=1}^{N}
c_k |a_k\rangle
\otimes |\gamma_k\rangle
\otimes |\theta_k\rangle,
\end{equation}
where $
|\gamma_k\rangle
=
e^{-ig_p a_kp_{\rm P}}|\gamma\rangle,$
 and $
|\theta_k\rangle
=
e^{-ig_m a_kp_{\rm MA}}|\theta\rangle .$

Post-selecting the measurement apparatus on the pointer outcome $x_m$
gives
\begin{equation}
\begin{aligned}
&P_{g_p}(x,x_m)
= ||(I \otimes \langle x| \otimes \langle x_m|)|\Psi_2\rangle||^2 \\
&=\sum_{k=1}^{N}
|c_k|^2
|\gamma(x-g_pa_k)|^2
|\theta(x_m-g_ma_k)|^2,
\end{aligned}
\end{equation}
where we used the orthogonality of $|a_k\rangle$. Consequently,
\begin{equation}
P_{g_p}(x_m)
=
\sum_{k=1}^{N}
|c_k|^2
|\theta(x_m-g_ma_k)|^2.
\end{equation}
After substituting in Eq.~\eqref{eq:composite_finite_strength_value}, one obtains
\begin{equation}
{}_{x_m}\langle A\rangle_{\psi}^{(g_p)}
=
\frac{
\displaystyle
\sum_{k=1}^{N}|c_k|^2 a_k
|\theta(x_m-g_ma_k)|^2
}{
\displaystyle
\sum_{k=1}^{N}|c_k|^2
|\theta(x_m-g_ma_k)|^2
},
\end{equation}
which is independent of $g_p$.
\end{proof}

The apparatus strength $g_m$ remains in the conditional weights through
the translated wavepackets $\theta(x_m-g_m a_k)$. In the limit
$g_m\rightarrow 0$, the apparatus carries no information that distinguishes
the eigenvalues of $A$, and therefore
${}_{x_m}\langle A\rangle_{\psi}^{(g_p)}
\longrightarrow
\langle A\rangle_{\psi}$. In the opposite regime, when the apparatus wavepackets associated with
different eigenvalues become well separated, the outcome $x_m$ identifies a
particular eigenvalue and the conditional value approaches that eigenvalue.

This is the continuous-variable analogue of the precision invariance obtained
in Eq.~\eqref{eq:finite_strength_Z_value}: when the probe and the measurement apparatus couple through the same
system observable, conditioning on the apparatus outcome modifies the
relative weights associated with the eigenvalue contributions without
introducing any dependence on the probe strength.

\subsection{Distinct Entangling Observables}

We now consider the more general case where $O_{\rm P} = B$ and $O_{\rm MA} = A$ where $B = \sum_{j = 1}^N  b_j |b_j\rangle \langle b_j| \neq A$ is another nondegenerate observable. We investigate conditions under which probe-strength invariance is preserved.
\newtheorem*{theorem}{Theorem}
\newtheorem*{corollary}{Corollary}

\begin{theorem}

Let an $N$-dimensional system initially prepared in
$|\psi\rangle
=
\sum_{k=1}^{N}c_k|b_k\rangle$
interact successively with a probe and a measurement apparatus through
\begin{equation*}
U_{\rm P}
=
e^{-ig_p B\otimes p_{\rm P}},
\qquad
U_{\rm MA}
=
e^{-ig_m A\otimes p_{\rm MA}}.
\end{equation*}
Suppose that, following both interactions, the system is post-selected onto the eigenstate $| b_f\rangle$ and that the corresponding post-selection probability is nonzero. 

Then the finite-strength conditional value ${}_{b_f}\langle B\rangle_{\psi}^{(g_p)}$ is
independent of $g_p$ when $N  = 2$ and the eigenbases of $A$ and $B$ are mutually unbiased (MUB). 
\end{theorem}

\begin{proof}

In the $N$-dimensional case, after the $U_{\rm P}$ and $U_{\rm MA}$ interactions, the state is

\begin{equation}
|\Psi_2\rangle
=
\sum_{k,r}
c_k
\langle a_r|b_k\rangle
|a_r\rangle \otimes
|\gamma_k\rangle \otimes
|\theta_r\rangle ,
\end{equation}
with $
|\theta_r\rangle
=
e^{-ig_m a_r p_{\rm MA}}|\theta\rangle$ and $|\gamma_k\rangle =
e^{-ig_p b_k p_{\rm P}}|\gamma\rangle.$
Now post-select the system onto $|b_f\rangle$. The unnormalized post-selected probe-apparatus state is
\begin{equation}
|\Phi_f\rangle
=
(\langle b_f| \otimes I \otimes I)|\Psi_2\rangle
=
\sum_{k=1}^{N}
c_k |\gamma_k\rangle \otimes |\chi_{f k}\rangle ,
\end{equation}
where
$|\chi_{f k}\rangle
=
\sum_{r=1}^{N}
\langle b_f|a_r\rangle
\langle a_r|b_k\rangle
|\theta_r\rangle .$
Hence the post-selection probability is
\begin{equation}
P_{g_p}(b_f)
=
\langle \Phi_f|\Phi_f\rangle
=
\sum_{k,k'=1}^{N}
c_k c_{k'}^*
G_{k'k}
K^{(f)}_{k'k},
\end{equation}
where 
$G_{k'k}
=
\langle \gamma_{k'}|\gamma_k\rangle
$ and 
$K^{(f)}_{k'k}
=
\langle \chi_{f k'}|\chi_{f k}\rangle .$

Similarly, the numerator of the finite-strength conditional value is
\begin{equation}
\begin{aligned}
&\frac{1}{g_p} \int dx\, x\, P_{g_p}(x,b_f) = \frac{1}{g_p}
\langle \Phi_f|x_{\rm P}|\Phi_f\rangle\\
&=
\frac{1}{g_p}
\sum_{k,k'=1}^{N}
c_k c_{k'}^*
\langle \gamma_{k'}|x_{\rm P}|\gamma_k\rangle
K^{(f)}_{k'k}.
\end{aligned}
\end{equation}
 Since $G_{kk}=\langle\gamma_k|\gamma_k\rangle=1$ and $\langle\gamma_{k}|x_{\rm P}|\gamma_k\rangle = g_p b_k$, we observe that if we impose the sufficient condition 
 \begin{equation}   
K^{(f)}_{k'k}=0 \qquad
\text{for all } k'\neq k
\label{eq:condition}
 \end{equation}
then the conditional probe average becomes
\begin{equation}
{}_{b_f}\langle B\rangle_{\psi}^{(g_p)}
=
\frac{
\sum_{k=1}^{N}
|c_k|^2 b_k K^{(f)}_{kk}
}{
\sum_{k=1}^{N}
|c_k|^2 K^{(f)}_{kk}
}, \label{eq:cond_avg_B}
\end{equation}
which is independent of $g_p$.

It remains to show that the condition from Eq.~\eqref{eq:condition} holds for $N=2$ when the $A$- and $B$-bases are mutually unbiased.
Let the eigenstates of $B$ be parametrized relative to the $A$ eigenbasis as
\begin{equation}
|b_\pm\rangle
=
\cos\beta\,|a_\pm\rangle
\pm e^{\pm i\phi}\sin\beta\,|a_\mp\rangle,
\end{equation}
where $\beta\in[0,\pi/2]$ and $\phi\in[0,2\pi)$. The $A$-interaction generates the apparatus states
$|\theta_\pm\rangle=e^{-ig_m a_\pm p_{\rm MA}}|\theta\rangle$.
Since the initial pointer wavefunction $\theta(x)$ is real, their overlap
$\langle\theta_+|\theta_-\rangle$ is also real.

After post-selecting onto
$|b_+\rangle$, the off-diagonal element is
\begin{equation}
K^{(+)}_{-+}
=
\langle\chi_{+-}|\chi_{++}\rangle
=
-\frac{1}{4}e^{i\phi}\sin(4\beta)
\left(1-\langle\theta_+|\theta_-\rangle\right).
\end{equation}
For $\beta=\pi/4$, corresponding to mutually unbiased eigenbases, the off-diagonal elements vanish,
\begin{equation}
K^{(+)}_{-+}
=
K^{(+)}_{+-}
=
0.
\end{equation}

The same cancellation occurs for $\beta=0$, where the $A$ and
$B$ eigenbases coincide. Thus, probe-strength invariance also holds for system
post-selection in the identical-basis limit, complementing the
measurement-apparatus post-selection considered in the previous Lemma.

Returning to the mutually unbiased case, $\beta=\pi/4$, the corresponding
diagonal elements reduce to
\begin{equation}
K^{(+)}_{++}
=
\frac{1+\mu}{2},
\qquad
K^{(+)}_{--}
=
\frac{1-\mu}{2},
\qquad
\mu\equiv
\langle\theta_+|\theta_-\rangle .
\end{equation}

Substitution into Eq.~\eqref{eq:cond_avg_B} gives
\begin{equation}
{}_{b_+}\!\langle B\rangle_{\psi}^{(g_p)}
=
\frac{
|c_+|^2 b_+(1+\mu)
+
|c_-|^2 b_-(1-\mu)
}{
|c_+|^2(1+\mu)
+
|c_-|^2(1-\mu)
}.
\end{equation}
\end{proof}
The result depends on the apparatus strength through $\mu$, but is
exactly independent of the probe coupling $g_p$. It is straightforward to check that
post-selection on $|b_-\rangle$ also yields invariance under $g_p$.
Thus, for $N=2$, mutual unbiasedness of the $A$- and $B$-bases is
sufficient to eliminate the probe-coherence terms and hence to
guarantee exact measurement-strength invariance.

This is precisely the structure encountered in the disturbance protocol of the
Lund--Wiseman qubit model, where the probe and measurement apparatus couple
through the mutually unbiased observables $X$ and $Z$, respectively. However,
this cancellation relies on the special geometry of a two-dimensional system.
For $N>2$, mutual unbiasedness does not in general enforce the orthogonality of
the post-selected apparatus states $|\chi_{fk}\rangle$ required by
Eq.~\eqref{eq:condition}, nor does it by itself guarantee probe-strength
invariance. An explicit qutrit counterexample with mutually unbiased
$A$- and $B$-eigenbases is given in Appendix~\ref{app:qutrit_counterexample}.

Importantly, however, $N=2$ and mutual unbiasedness are not necessary conditions for probe-strength invariance in a larger Hilbert space.
\begin{corollary}
Let the system have even dimension $N=2m$, and suppose that the Hilbert space
admits a common invariant decomposition
\begin{equation}
\mathcal H_{2m}
=
\bigoplus_{q=1}^{m}\mathcal H_q^{(2)}
\end{equation}
for the observables $A$ and $B$. If, within each two-dimensional subspace
$\mathcal H_q^{(2)}$, the eigenbases of $A$ and $B$ are mutually unbiased,
then, under the assumptions of the preceding theorem, the finite-strength
conditional value
${}_{b_f}\langle B\rangle_{\psi}^{(g_p)}$
is independent of the probe strength $g_p$ for any post-selection onto a
$B$ eigenstate $|b_f\rangle$ with nonzero probability.
\end{corollary}

\begin{proof}
Let $|b_f\rangle\in\mathcal H_q^{(2)}$. Since $A$ and $B$ leave each
$\mathcal H_q^{(2)}$ invariant,
\begin{equation}
\langle b_f|a_r\rangle\langle a_r|b_k\rangle=0
\end{equation}
whenever $|b_k\rangle\notin\mathcal H_q^{(2)}$. Hence
$|\chi_{fk}\rangle=0$ for all $|b_k\rangle$ outside the sector containing
$|b_f\rangle$.

The post-selected state therefore receives contributions only from the two
$B$ eigenstates in $\mathcal H_q^{(2)}$. Within this sector, the $A$- and
$B$-eigenbases are mutually unbiased, so the preceding theorem gives
\begin{equation}
K^{(f)}_{k'k}=0,
\qquad k'\neq k .
\end{equation}
Thus the sufficient condition in Eq.~\eqref{eq:condition} holds, and
Eq.~\eqref{eq:cond_avg_B} is independent of $g_p$.
\end{proof}

The theorem and corollary therefore distinguish between genuinely
higher-dimensional dynamics and higher-dimensional systems that decompose
into effective two-dimensional sectors. Mutual unbiasedness guarantees the
cancellation for a qubit, but does not in general do so for $N>2$.
Nevertheless, the same mechanism survives in higher dimensions whenever the
relevant dynamics and post-selection are confined to invariant
two-dimensional subspaces. Thus, probe-strength invariance is controlled not
by Hilbert-space dimension alone, but by the structure of the observables and
the post-selections explored by the protocol.

\section{Conclusion}

We have shown that the weak-measurement limit is not always necessary for the operational reconstruction of weak values. Within the qubit model of the Lund--Wiseman measurement--disturbance protocol, the finite-strength conditional values reconstructed from the probe are exactly independent of the probe strength and coincide with their weak values.

Our analysis identifies the origin of this invariance in the combined structure of the probe interaction, the subsequent measurement-apparatus interaction, and the final post-selection. When both interactions are generated by the same observable, probe-strength invariance holds for the post-selections considered here in arbitrary system dimension. When the interactions are generated by distinct observables, mutual unbiasedness of their eigenbases is sufficient to enforce the required cancellation for a two-dimensional system. This accounts for the otherwise surprising invariance of the Lund--Wiseman disturbance protocol, where the probe and measurement apparatus couple through the mutually unbiased observables $X$ and $Z$.

Mutual unbiasedness alone does not generally guarantee strength invariance, although the phenomenon persists in larger Hilbert spaces containing invariant two-dimensional sectors. A natural next step is to determine whether higher-dimensional systems admit other mechanisms for invariance and to classify the measurement interactions and post-selections for which it occurs.

Moreover, Eq.~\eqref{eq:cond_avg_B} shows that, under the cancellation
responsible for probe-strength invariance, the conditional value reduces to
a convex combination of the eigenvalues of $B$ and is therefore confined to
its spectral range. Anomalous weak values, by contrast, lie outside this
range. Thus, the mechanism producing exact strength invariance identified
here simultaneously excludes anomalous weak values. Whether this connection
between strength invariance and the absence of anomalous weak values extends
beyond the present class of protocols remains an open question.

Continuous-variable systems provide a natural setting in which to investigate
this question further. Weak-valued probability distributions have, for
example, been employed to characterize momentum transfer in which-way
measurements~\cite{garretson2004uncertainty}, where anomalous
quasiprobabilities can arise. Determining whether any analogue of the
finite-strength invariance identified here can coexist with such anomalous
behaviour would therefore provide a useful test of how general the connection
between strength invariance and non-anomalous conditional values is.

\begin{acknowledgments}
We thank Giuseppe Di Pietra, Lodovico Scarpa, Nicetu Tibau Vidal, and Aditya Iyer for valuable discussions and feedback on this work. C.A. acknowledges support from the Clarendon Fund. V.V. thanks the Gordon and Betty Moore Foundation as well as the Conjecture Institute for their financial support. GPT-5.6 Sol was used during the preparation of this manuscript solely to assist with grammar, wording, and sentence structure. All conceptual content, arguments, and results were developed by the authors. AI-assisted edits were reviewed and verified by the authors, who take full responsibility for the manuscript.
\end{acknowledgments}

\newpage

\appendix
\begin{widetext}
\section{Finite-Strength Disturbance}
\label{app:finite_strength_disturbance}
For the disturbance protocol, the probe couples to the system in the $X$
basis, whereas the measurement apparatus couples in the $Z$ basis. The probe
is subsequently read out in the $Z$ basis, yielding $Z_p$, while the final
system readout is performed in the $X$ basis, yielding $X_s$.

Writing 
\begin{equation}
|\psi\rangle=c_{+}|+\rangle+c_{-}|-\rangle,
\qquad
c_{\pm}=\frac{\alpha\pm\beta}{\sqrt{2}},
\end{equation}
the finite-strength conditional value is
\begin{equation}
{}_{X_s}\langle X\rangle_{\psi}^{(s_\gamma)}
=
\frac{P(Z_p=+1|X_s)-P(Z_p=-1|X_s)}{s_\gamma}.
\end{equation}
As in the precision calculation, the required probabilities are
\begin{equation}
P(Z_p,X_s)
=
\left\|
\langle X_s|\otimes\langle Z_p|\otimes I
\left( U^{\rm app}_{\rm Z}U^{\rm probe}_{\rm X} \right)
|\psi\rangle \otimes |\gamma\rangle \otimes|\theta\rangle
\right\|^2 .
\end{equation}
For $X_s=+1$, direct evaluation gives
\begin{align}
P(Z_p=+1,X_s=+1)
&=
\frac{1}{2}\left[
\gamma^2|c_+|^2(1+\sin(2\theta))
+\bar{\gamma}^2|c_-|^2(1-\sin(2\theta))
\right],\\
P(Z_p=-1,X_s=+1)
&=
\frac{1}{2}\left[
\bar{\gamma}^2|c_+|^2(1+\sin(2\theta))
+\gamma^2|c_-|^2(1-\sin(2\theta))
\right],
\end{align}
and hence
\begin{equation}
P(X_s=+1)
=
\frac{1}{2}\left[
|c_+|^2(1+\sin(2\theta))
+
|c_-|^2(1-\sin(2\theta))
\right].
\end{equation}

Therefore,
\begin{equation}
{}_{X_s=+1}\langle X\rangle_{\psi}^{(s_\gamma)}
=
\frac{
|c_+|^2(1+\sin(2\theta))-|c_-|^2(1-\sin(2\theta))
}{
|c_+|^2(1+\sin(2\theta))+|c_-|^2(1-\sin(2\theta))
}.
\label{eq:app_Xplus}
\end{equation}

Similarly, for $X_s=-1$,
\begin{equation}
{}_{X_s=-1}\langle X\rangle_{\psi}^{(s_\gamma)}
=
\frac{
|c_+|^2(1-\sin(2\theta))-|c_-|^2(1+\sin(2\theta))
}{
|c_+|^2(1-\sin(2\theta))+|c_-|^2(1+\sin(2\theta))
}.
\label{eq:app_Xminus}
\end{equation}

Thus, for either post-selection, the probe strength $s_\gamma$ cancels
exactly:
\begin{equation}
{}_{X_s}\langle X\rangle_{\psi}^{(s_\gamma)}
=
{}_{X_s}\langle X\rangle_{\psi}^{\rm weak}.
\end{equation}
Consequently, the weak-valued probabilities entering the reconstruction
of $\eta(X)$ are also independent of the probe strength.
\newpage 
\section{Qutrit Counterexample}
\label{app:qutrit_counterexample}

We give an explicit $N=3$ counterexample showing that mutual unbiasedness
alone does not guarantee probe-strength invariance. Let
\begin{equation}
A=\sum_{r=0}^{2}a_r|a_r\rangle\langle a_r|,
\qquad
(a_0,a_1,a_2)=(-1,0,1),
\end{equation}
and define the $B$ eigenbasis by
\begin{equation}
|b_k\rangle
=
\frac{1}{\sqrt{3}}
\sum_{r=0}^{2}\omega^{rk}|a_r\rangle,
\qquad
\omega=e^{2\pi i/3},
\end{equation}
with eigenvalues $(b_0,b_1,b_2)=(-1,0,1)$. It is straightforward to verify that the eigenbasis is orthonormal by using the identity $\sum_{r = 0}^{N-1} e^{2\pi i r(m-n)/N} =N \delta_{m,n}$. Moreover
$|\langle a_r|b_k\rangle|^2=1/3$, so the two bases are mutually unbiased.

Take a real centered Gaussian apparatus pointer and post-select on
$|b_0\rangle$.
The post-selected apparatus overlaps satisfy
\begin{equation}
K^{(0)}_{00}=\frac{3+4\mu+2\mu^4}{9},
\qquad
K^{(0)}_{11}=\frac{3-2\mu-\mu^4}{9},
\qquad
\operatorname{Re}K^{(0)}_{01}
=
\frac{\mu^4-\mu}{18}
\end{equation}
with $\mu=e^{-g_m^2/(8\sigma_m^2)}$. Hence $K^{(0)}_{01}\neq0$ for any nonzero finite apparatus coupling.

Now choose
\begin{equation}
|\psi\rangle
=
\frac{|b_0\rangle+|b_1\rangle}{\sqrt{2}},
\end{equation}
and a real centered Gaussian probe pointer. Its relevant overlap is
\begin{equation}
G(g_p)
=
\langle\gamma_0|\gamma_1\rangle
=
e^{-g_p^2/(8\sigma_p^2)}.
\end{equation}
The conditional value then becomes
\begin{equation}
{}_{b_0}\langle B\rangle_{\psi}^{(g_p)}
=
-
\frac{
K^{(0)}_{00}
+
G(g_p)\operatorname{Re}K^{(0)}_{01}
}{
K^{(0)}_{00}
+
K^{(0)}_{11}
+
2G(g_p)\operatorname{Re}K^{(0)}_{01}
}.
\label{eq:qutrit_counterexample}
\end{equation}
Since $\operatorname{Re}K^{(0)}_{01}\neq0$, the conditional value depends
explicitly on $g_p$. Thus, already for $N=3$, mutually unbiased $A$- and
$B$-eigenbases do not in general guarantee probe-strength invariance.

\end{widetext}

\bibliography{bibs}

\end{document}